\documentclass[lettersize,journal]{IEEEtran}
\IEEEoverridecommandlockouts
\usepackage{amsmath,amsthm,amsfonts,amssymb,amsbsy,bm,paralist,color,cite}
\usepackage{graphicx,algorithmic,algorithm,relsize}
\usepackage{multicol}
\usepackage{multirow}
\usepackage{caption}
\usepackage{epstopdf}
\usepackage{textcomp}
\usepackage{xcolor}
\usepackage{subfigure}
\usepackage[font=scriptsize]{caption} 

\newcommand{\om}{\bm{\omega}}
\newcommand{\h}{\bm{h}}

\newcommand{\bt}{\bar{\tau}}
\newcommand{\bw}{\bm{W}}
\newcommand{\bz}{\bm{z}}
\newcommand{\bA}{\bm{A}}
\newcommand{\bB}{\bm{B}}
\newcommand{\be}{\bm{\epsilon}}
\newcommand{\tRe}{\textnormal{Re}}
\newcommand{\hath}{\hat{\bm{h}}}
\newcommand{\bep}{\bm{\epsilon}}
\newcommand{\hR}{\hat{R}}

\makeatletter
\newtheorem{proposition}{Proposition}
\newtheorem{Lemma}{Lemma}
\newtheorem{Remark}{Remark}
\begin{document}
	\title{Joint Robust Beamforming Design for WPT-assisted D2D Communications in MISO-NOMA: Fractional Programming and Deep Reinforcement Learning \\
		{\footnotesize }
		\author{Shiyu Jiao, Fang Fang,~\IEEEmembership{Member,~IEEE} and Zhiguo Ding,~\IEEEmembership{Fellow,~IEEE}\vspace{-0.9cm}
			\thanks{Shiyu Jiao and Zhiguo Ding are with School of Electrical and Electronic Engineering, The University of Manchester, M13 9PL, U.K. (e-mail: shiyu.jiao@manchester.ac.uk and zhiguo.ding@manchester.ac.uk).
				
				Fang Fang is with the Department of Electrical and Computer Engineering and the Department of Computer Science, Western University, London, ON N6A 3K7, Canada. (e-mail: fang.fang@uwo.ca).}}}
	
	\maketitle	
	\begin{abstract}
		This paper proposes a scheme for the envisioned sixth-generation (6G) ultra-massive Machine Type Communications(umMTC). In particular, wireless power transfer (WPT) assisted communication is deployed in non-orthogonal multiple access (NOMA) downlink networks to realize spectrum and energy cooperation. This paper focuses on joint robust beamforming design to maximize the energy efficiency of WPT-assisted D2D communications in multiple-input single-output (MISO)-NOMA downlink networks. To efficiently address the formulated non-concave energy efficiency maximization problem, a pure fractional programming (PFP) algorithm is proposed, where the time switching coefficient of the WPT device and the beamforming vectors of the base station are alternatively optimized by applying the Dinkelbach method and quadratic transform respectively. To prove the optimality of the proposed algorithm, the partial exhaustive search algorithm is proposed as a benchmark. A deep reinforcement learning (DRL)-based method is also applied to directly solve the non-concave problem. The proposed PFP algorithm and the DDPG-based algorithm are compared in the presence of different channel estimation errors. Simulation results show that the proposed PFP algorithm outperforms the DDPG-based algorithm if perfect channel state information (CSI) can be obtained or just have minor errors, while the DDPG-based algorithm is more robust when the channel estimation accuracy is unsatisfactory. On the other hand, one can conclude that the NOMA scheme can provide a higher gain than OMA on the energy efficiency of the WPT-assisted D2D communication in legacy multi-user downlink networks. 
	\end{abstract}
	\begin{IEEEkeywords}
		non-orthogonal multiple access (NOMA), wireless power transfer (WPT),  Device-to-Device (D2D), convex optimization, deep reinforcement learning (DRL)
	\end{IEEEkeywords}	
	\section{Introduction}
	With the development of wireless communication from fifth-generation (5G) to sixth-generation (6G), the demand for massive machine-type communications (mMTC) is raised to ultra-mMTC (umMTC) \cite{zhang20196g}. The emergence of new usage scenarios and applications, such as the Internet of Things (IoT), dramatically drove this upgrade. However, simultaneously serving massive devices by utilizing the limited spectrum resource is challenging. In the meanwhile, the ultra-dense networks formed by massively connected devices lead to huge power consumption, which significantly increases the operating cost of wireless communication networks. Thus, a spectrum and energy efficient solution that enables ultra-dense networks is urgent and critical.  
	
	To support ultra-dense networks in 5G and 6G, D2D communication has still been regarded as a promising scheme and will be gradually appended to existing cellular networks \cite{hashim2019ultra,gismalla2022survey,duong2019ultra}. D2D communication was introduced in 4G LET as a kind of peer-to-peer short wireless transmission between \texttt{}devices without relaying by base stations (BS) or access points (AP)\cite{zhang2020envisioning, jayakumar2021review}, which can mitigate the load on the BSs. Generally, D2D communication is classified into two categories: Inband D2D and Outband D2D\cite{jayakumar2021review}. Inband D2D communication utilizes the same licensed spectrum in cellular networks with cellular devices such as mobile phones. For outband D2D communication, it occurs in Ad-hoc networks such as Wi-Fi, Bluetooth etc., which is out of the scope of this paper. In terms of licensed spectrum utilization, there are two ways to assign the spectrum to D2D devices, namely Underlay and Overlay \cite{elsawy2014analytical,pei2013resource}. The Underlay type allows the licensed spectrum to be shared with both D2D devices and original cellular devices while the Overlay type divides the licensed spectrum into two parts and allocates them to cellular devices and D2D devices respectively. Although D2D has been widely studied in existing works \cite{lin2014overview,asadi2014survey,jameel2018survey}, it still has many challenges that demand prompt solutions \cite{gismalla2022survey, zhang2020envisioning}. For example, if deploying battery-powered D2D pairs in a legacy cellular network, resource allocation, interference controlling and energy efficiency improvement, etc., are required to enhance the network's performance and prolong the D2D pair's battery life.

	To further improve the spectrum efficiency of the D2D pair, non-orthogonal multiple access (NOMA) can be applied in cellular networks \cite{kimy2013non, ding2015application, ding2014performance}. In 6G, NOMA remains in the spotlight and is ever-evolved in academia and industry. In recent studies on using NOMA, the authors of \cite{liu2021application} and \cite{yuan2021noma} demonstrated its enormous potential in 6G and tremendous benefit for 6G. Furthermore, the spectral efficiency of NOMA enabled IoT network for 6G was further improved by \cite{khan2020spectral}. Different from conventional multiple access techniques, including frequency-division multiple access (FDMA), time-division multiple access (TDMA), code-division multiple access (CDMA) and orthogonal frequency division multiple access (OFDMA) for previous generations of cellular communications, NOMA allows all users to share the same frequency band and channel coding at the same time. By applying NOMA, high mutual interference will be introduced when the NOMA users are decoding signals. Thus, successive interference cancellation (SIC) is applied at the receiver \cite{saito2015performance}.
	
	For battery-powered devices, energy is one of the most precious resources. Thus, how to save their energy and/or improve their energy efficiency are emergent and important. This motivates the use of wireless power transfer (WPT) in this paper. The core idea of WPT in wireless communications is that WPT-enabled devices harvest energy from radio frequency (RF) signals. Generally, there are two types of WPT, i.e., time switching (TS) and power splitting (PS) WPT respectively \cite{zhang2013mimo}. In particular, for TS-WPT, the receiver periodically switches between harvesting energy mode and transmitting signals or decoding information mode \cite{choi2019toward,ding2015application2,lu2014wireless}, whereas the PS-WPT receiver splits the received signal into two power level streams and then assigns them to the energy harvesting receiver and information decoding receiver respectively \cite{cao2018analysis,ye2017power}. Note that, in this paper, we only consider the user of TS-WPT.
	\subsection{Related Works}
	In literature, D2D, NOMA and WPT were combined and studied in pairs or all together for different scenarios. The authors in \cite{ding2021harvesting} analysed the performance of a NOMA uplink network consisting of a single non-energy-constrained device and multiple energy-constrained WPT supported devices, which provides the research directions for WPT-NOMA. The authors in \cite{goktas2022wireless} maximized the uplink sum rate of multiple WPT-assisted devices in a single user downlink NOMA network, where the time switching coefficient and power allocation were alternatively optimized. In \cite{tang2019energy}, the energy efficiency of a downlink SWIPT-enabled NOMA system with TS-based terminals was maximized by jointly optimizing the time switching coefficients of terminals and the power allocation strategy of the BS. The authors in \cite{yu2021optimal} obtained the optimal power allocation scheme for a single-carrier single-uplink-user NOMA-enabled network by using convex optimization, where one D2D transmitter and two D2D receivers are taken into account. \cite{pan2017resource} optimized the resource allocation and channel assignment scheme in a NOMA downlink cellular network, where multiple D2D devices are deployed. A recent work \cite{khazali2021energy} applied WPT to two NOMA uplink users groups to improve energy efficiency and spectrum efficiency, where users in these two groups perform energy harvesting and signal transmission alternatively.
	\subsection{Motivation, Challenges and Contributions}
	To address the spectral and energy challenges in the 6G ultra-dense networks as aforementioned and motivated by the green communications attribute of WPT \cite{xue2022research,zhao2021special,allamehzadeh2021wireless}, naturally, we combine WPT, D2D communication and NOMA technique to realize spectrum and energy cooperation. In particular, the WPT-assisted Inband-Underlay D2D communication is deployed into NOMA downlink networks, of which WPT-assisted devices can harvest energy from NOMA downlink signal and can share the same spectrum with NOMA downlink users. Since energy efficiency is a very important performance index for WPT devices, this paper aims to improve the energy efficiency of the WPT-assisted D2D communication. Due to the fraction form of the problem, fractional programming is naturally selected. On the other hand, inspired by many works on applying learning-based methods to communication optimization problems, for example, \cite{zappone2018online} has performed the online energy-efficient power control in wireless networks by deep neural networks, this paper uses a deep reinforcement learning (DRL)-based approach to solve the problem. 
	
	To realize the spectrum and energy cooperation aforementioned, there are challenges that need to be overcome \cite{zhang2020envisioning}. First, due to the spectrum sharing protocol, severe co-channel interference to NOMA downlink users will be introduced when the D2D devices are appended to cellular networks. Therefore, interference control has to be carried out to guarantee the original cellular users' quality of service (QoS) when D2D devices are deployed. Second, the inaccuracy of channel estimation in practical systems makes the beamforming more challenging to design, and hence the robust beamforming design is important. 
	
	Different from \cite{pei2018energy} maximizing the energy efficiency of a D2D pair in a single-antenna BS NOMA uplink network, this paper maximizes the energy efficiency of the WPT-assisted D2D communication in a multiple-input single-output (MISO)-NOMA downlink system. The main contributions are summarised as follows:
	\begin{itemize}
		\item In this paper, we propose a novel scheme that can realize spectrum and energy cooperation, where a WPT-enabled D2D pair is inserted into a multi-user MISO-NOMA downlink network. Assume that the D2D transmitter adopts the harvesting energy then transmitting information strategy. This scenario can be extended to the legacy user-clustered hybrid NOMA networks \cite{you2020user} and is bound to appear in the process of future cellular networks upgrade to ultra-dense networks. 
		\item The formulated energy efficiency maximization problem is not concave and the two variables (beamforming vectors of the BS and time switching coefficient of the WPT device) are highly coupled in both the fractional objective function and constraints. To efficiently solve the non-concave problem, it first is simplified, and then an alternating algorithm, namely, pure fractional programming (PFP) is proposed. Specifically, the proposed problem is split into two subproblems to decouple the coupled variables. Afterwards, the time switching coefficient is optimized by applying the Dinkelbach method. For robust beamforming designing, the multi-dimension complex quadratic transform is used. Simulation results reveal that the proposed algorithm can converge perfectly for different schemes and channel assumptions. 
		\item A partial exhaustive search algorithm which can bypass the alternating operation is proposed as a benchmark to verify the PFP algorithm's optimality. Additionally, a deep reinforcement learning (DRL) approach (i.e., deep deterministic policy gradient (DDPG)) is applied to directly solve the non-concave energy efficiency maximization problem. Simulation results reveal a fascinating finding: the proposed PFP algorithm can provide better performance when channel estimation is accurate (perfect CSI or only minor error exists), while the DDPG-based algorithm has the capability to mitigate the adverse impact caused by channel estimation error.
		\item Simulations are also performed for the scenario when orthogonal multiple access (OMA) is applied. Simulation results illustrate that the proposed algorithm is also applicable to OMA. However, one can conclude is, in the examined system model, the WPT-assisted D2D communication in  MISO-NOMA downlink networks can obtain higher energy efficiency enhancement than that in MISO-OMA downlink networks.
	\end{itemize} 
	\subsection{Organization}
	The rest of this paper is arranged as follows. Section II describes the system model as well as the energy harvesting and information transmission strategy. The problem formulation and preliminary handling are discussed in section II. In section III, solutions to two subproblems, the proposed PFP algorithm and the partial exhaustive search algorithm are provided. Section IV introduces the DDPG algorithm and discusses its application to the original non-concave problem. In section V, the detailed deep neural networks structure and parameters, simulation parameters and simulation results are provided and analysed. In the end, we conclude this paper in section VII.  
	\section{System model and Problem formulation}
	\subsection{System Model}
	\begin{figure}[t]   
		\centering
		\includegraphics[width=0.8\linewidth]{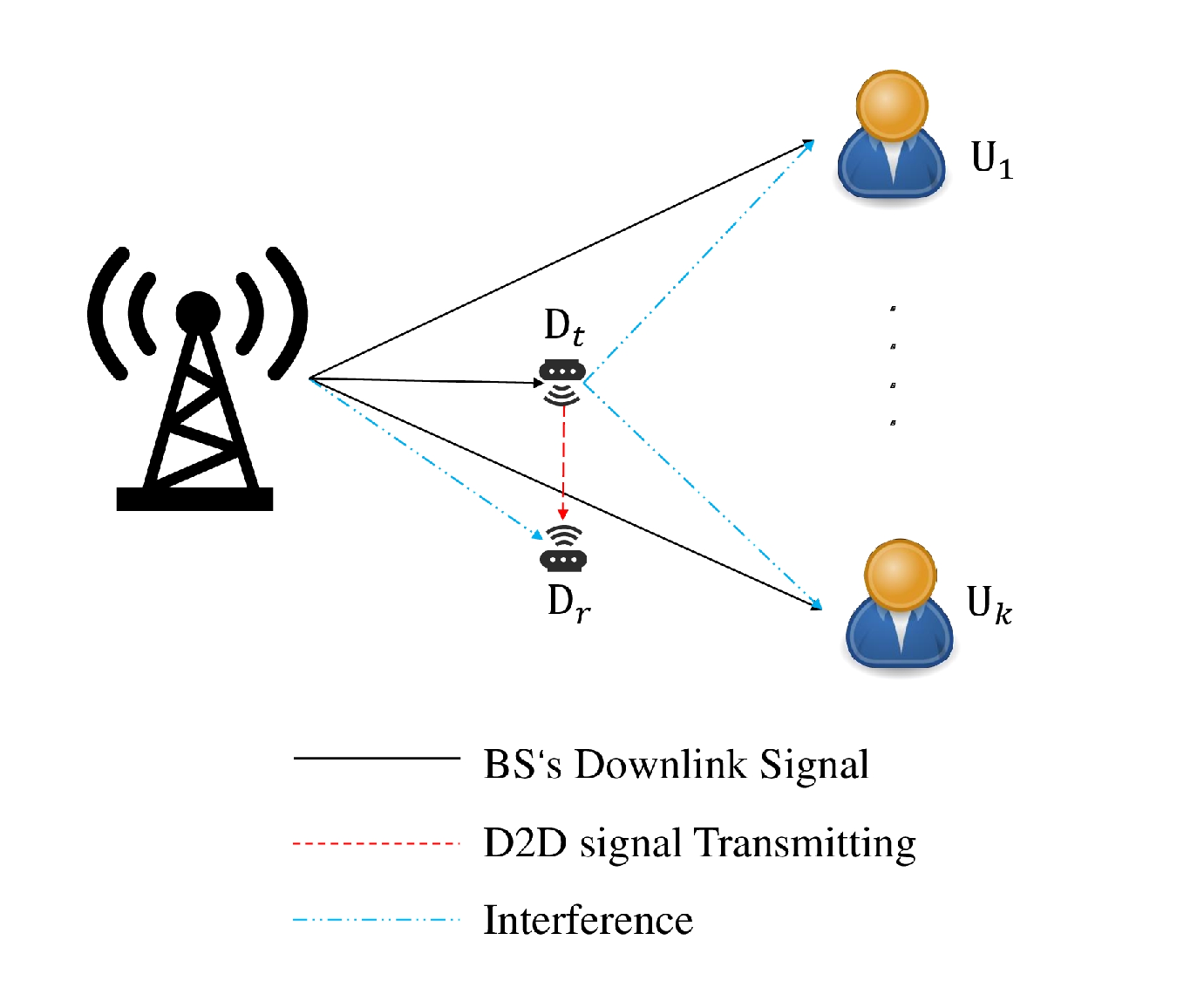}\\
		\caption{System model.}\label{system_model}
	\end{figure} 
	
	Consider a MISO-NOMA downlink network with WPT-assisted D2D communications as shown in Fig. \ref{system_model}. This network consists of a base station (BS), $K$ NOMA downlink users and a pair of D2D devices which are denoted by $D_t$ (the D2D signal transmitter) and $D_r$ (the D2D signal receiver) respectively, where the BS is equipped with $M$ antennas while all other nodes are equipped with a single antenna. This system can further be extended to the legacy user-clustered hybrid NOMA networks \cite{you2020user}. Assume that $D_t$ is a WPT-assisted device and apply the harvest-energy-then-transmit strategy. In detail, $D_t$ performs the energy harvesting and then stores it during the first $\tau T$ seconds by utilizing the BS transmitted downlink signal. During the rest $(1-\tau)T$ seconds, $D_t$ sends its signal $s_D$ to $D_r$ by using the harvested energy. 
	$\tau$ represents the time-switching coefficient ($0\leq \tau \leq 1$) and $T$ is the duration of one time slot. For simplicity, we set $T=1$ in this paper. 
	Downlink users required signals during $\tau$ and $1-\tau$ seconds are respectively denoted as $s_k^{(1)}$ and $s_k^{(2)}$, $E\{|s_k^{(1)}|^2\} =E\{|s_k^{(2)}|^2\} = 1$, where $E\{\cdot\}$ is the expectation operator. The add-on WPT-D2D pair is admitted to share the same communication resource with NOMA downlink users. Due to the channel estimation is not always perfect in practice, this paper model the channels with estimation errors as follows \cite{yoo2006capacity}:
	\begin{equation}
	\hath = \h + \bep,	
	\end{equation}
	where $\h$ is the estimated complex Gaussian channel with variance $\sigma^2$ and $\bep$ represents the channel estimation error with variance $\sigma_{\epsilon}^2$.  
	
	During the first stage (i.e., $\tau$ seconds), the BS transmits the superposition signal $s^{(1)} = \sum_{k=1}^{K} \om_k s_k^{(1)}$ to all $K$ downlink users while $D_t$ harvests energy from the BS with broadcast signal $s^{(1)}$. Therefore, $D_t$ does not interfere with downlink users receiving their signals during this stage. The $U_k$'s received signal at this stage is given by
	\begin{equation}
		y_k^{(\tau)} = \h_k^H \sum_{k=1}^{K} \om_k s_k^{(1)} + \bep_k^H \sum_{k=1}^{K} \om_k s_k^{(1)} + n_k,
	\end{equation}
	where $\h_k \in \mathbb{C}^{M \times 1}$ is the channel vector from the BS to the $U_k$, $\om_k \in \mathbb{C}^{M \times 1}$ is the beamforming vector for the $U_k$,  $\bep_k$ denotes the $U_k$'s channel estimation error and $n_k\thicksim \mathcal{CN} (0,\sigma^2)$ denotes the additive White Gaussian noise (AWGN).  
	
	In this paper, $U_1$ is defined as the weakest user whereas $U_K$ is the strongest user (i.e., $1\leq k<t\leq K$ for (\ref{rtktot})). In other words, the channel gains are sorted as $|\h_1|^2 \leq |\h_2|^ \leq \cdots \leq |\h_K|^2$. According to the SIC principle, a stronger user (who has better channel gain) can decode the signal of weaker users (who has worse channel gain). Denote $\kappa_k$ as the interference set when the signal of the $k$-th user is decoding. Therefore, the data rate of $U_t$ to decode $k$ weaker users is given by
	\begin{equation}\label{rtktot}
		R^{(\tau)}_{k \to t} = \tau \log(1+\frac{|\h_t^H\om_k|^2}{\sum_{j\in \kappa_k} |\h_t^H\om_j|^2+\Theta_t+\sigma^2}),
	\end{equation}
	 where $\Theta_t = \sum_{i = 1}^{K} |\bep_t^H\om_i|^2$ is the interference caused by channel estimation errors. After removing the weaker users' signal, $U_t$ can decode its own signal by simply treating other stronger users' signal as interference. Therefore, the data rate that $U_t$ to decode its own signal in this stage is given by
	\begin{equation}\label{rtktok}
		R^{(\tau)}_{t \to t} = \tau \log(1+\frac{|\h_t^H\om_t|^2}{\sum_{j\in \kappa_t} |\h_t^H\om_j|^2+\Theta_t+\sigma^2}).
	\end{equation}
	
	In the considered MISO-NOMA downlink system, assume that all the energy beams can be harvested and linear energy harvest strategy is applied \cite{chen2013energy}. Therefore, the received power at $D_t$ is given by  
	\begin{equation}
		P_r = \sum_{k=1}^{K} |\h_{D_t}^H\om_k|^2,
	\end{equation}
	where $\h_{D_t}$ is the channel vector from the BS to $D_t$. Denote the BS maximum transmit power by $P_{max}$. We have $\sum_{k=1}^{K}|\om_k|^2 \leq P_{max}$.
	Assume that the harvested energy will be totally used to transmit the signal, whereas the circuit needed energy is provided by the battery. Therefore, the transmit power of $D_t$ can be represented by
	\begin{equation}
		P_t = \frac{\eta \tau P_r  }{1-\tau}, 
	\end{equation}
	where $0\leq\eta\leq1$ is the RF energy conversion coefficient.
	
	During the second stage (i.e., the rest $1-\tau$ seconds), downlink users receive $s^{(2)}$ while $D_t$ transmits its signal to $D_r$. The transmitted signal from $D_t$ to $D_r$ is denoted by $s_D$. Since the D2D pair shares the same frequency band, the signal $s_D$ sent by $D_t$ will interfere with downlink users to receive the signal $s_k^{(2)}$. Therefore, the $U_k$'s received signal in the second stage is given by 
	\begin{equation} \label{yk1-tau}
		y_k^{(1-\tau)} = \sqrt{P_t}h_{dk}s_D +  \h^H_k\sum_{k=1}^{K}\om_ks_k^{(2)} + \bep_k^H \sum_{k=1}^{K} \om_k s_k^{(2)}+n_k,
	\end{equation}
	where $h_{dk}$ is the channel between the $D_t$ and the $U_k$. For the D2D pair, the received signal at the $D_r$ is given by
	\begin{equation}
		y_{D_r}^{(1-\tau)} = \sqrt{P_t}h_{dd}s_D + \sqrt{P_t}\epsilon_{dd}s_D+ \h^H_{D_r}\sum_{k=1}^{K}\om_ks_k^{(2)} + n_D,
	\end{equation}
	where $h_{dd}$ denotes the channel gain from $D_t$ to $D_r$, $\epsilon_{dd}$ represents the channel estimation error,  $\h_{D_r}\in\mathbb{C}^{M\times1}$ is the channel vector between the BS and $D_r$ and $n_D\thicksim \mathcal{CN} (0,\sigma^2)$ is the AWGN at $D_r$.
	
	Due to the double-fading effect and the fact that $D_t$'s transmitted signal is introduced unexpectedly, directly decoding the BS signal by treating $s_D$ as interference is an appropriate decoding strategy for downlink users. Hence, the data rate of $U_t$ to decode $U_k$s' signal can be written as
	\begin{equation} \label{r1tktot}
		\begin{aligned}
			R^{(1-\tau)}_{k \to t}& = (1-\tau) \times\\& \log(1+\frac{|\h_t^H\om_k|^2}{P_t|h_{dt}|^2+\sum_{j\in \kappa_k}|\h_t^H\om_j|^2+\Theta_t+\sigma^2}),
		\end{aligned}
	\end{equation}
	where $h_{dt}$ denotes the channel gain between $D_t$ and the $U_t$.
	After removing those weaker users' signal, the downlink data rate of the $U_t$ to decode its own signal is 
	\begin{equation} \label{r1tktok}
		\begin{aligned}
			R^{(1-\tau)}_{t \to t} &= (1-\tau) \times \\& \log(1+\frac{|\h_t^H\om_t|^2}{P_t|h_{dt}|^2+\sum_{j\in \kappa_t} |\h_t^H\om_j|^2+\Theta_t+\sigma^2}).
		\end{aligned}
	\end{equation}
	For the D2D receiver, all the downlink users' signals from the BS are treated as interference because it is not practical for battery-powered devices to perform the SIC with high complexity. The data rate that $D_r$ to decode $s_D$ is given by
	\begin{equation}
		R_D = (1-\tau) \log(1+\frac{P_t|h_{dd}|^2}{\sum_{k = 1}^{K} |\h_{D_r}^H\om_k|^2+ P_t|\epsilon_{dd}|^2+\sigma^2} ).
	\end{equation}
	The downlink data rate of the $U_k$ to decode its own signal and the data rate of $U_t$ to decode weaker $U_k$'s signal during a whole time slot are respectively denoted by
	\begin{equation} \label{rk}
		R_k = R^{(\tau)}_{k \to k}+R^{(1-\tau)}_{k \to k},
	\end{equation}
	\begin{equation}
		R_{k\to t} = R^{(\tau)}_{k \to t}+R^{(1-\tau)}_{k \to t}.
	\end{equation}
	Denote the power to drive the $D_t$'s circuit by $P_c$. The total energy consumption of $D_t$ during one time slot can be represented as follows: 
	\begin{equation}\label{Ec}
		E_c = (1-\tau)P_t + P_c = \eta \tau P_r + P_c.
	\end{equation} 
	\subsection{Problem Formulation}
	For WPT devices, utilizing energy effectively is extremely essential. Hence, the aim of this paper is to maximize the energy efficiency of the WPT-assisted D2D pair while guaranteeing the NOMA downlink user's data rate by jointly optimizing the time switching coefficient $\tau$ and beamforming vectors $\om$. The optimization problem can be formulated as:
	\begin{subequations} \label{P1}
		\begin{align}
			\textnormal{P1:} \max_{\{ \tau, \om \}} \quad & \frac{R_D}{E_c} \label{P1a}\\
			\mbox{s.t.} \quad & \min\{R_{k \to t}, R_{k \to k}\} \geq R_{min}, 1\leq k \leq t \leq K \label{P1b}\\ 		
			\quad &  \sum_{k=1}^{K} |\om_k|^2 \leq P_{max}, 1\leq k\leq K  \label{P1c}\\
			\quad & 	0\leq \tau \leq 1, \label{P1d}
		\end{align}
	\end{subequations} 
	where $R_{min}$ denotes the minimum target date rate of NOMA downlink users. (\ref{P1b}) is the QoS constraint which also guarantees that downlink users can implement SIC successfully. (\ref{P1c}) is the total power constraint of the BS, and (\ref{P1d}) is to restrict the time-switching coefficient in the feasible range. $\tau = 0$ indicates that the D2D transmitter transmits the signal in the whole time slot. In contrast, $\tau = 1$ means the D2D transmitter harvests energy in the whole time slot. 
	
	\begin{Lemma}
		$R_{k \to k} \geq R_{min}$ is equivalent to $\frac{1}{1-\tau} R^{(1-\tau)}_{k \to k} \geq R_{min}$, and (\ref{Rkt>Rmin}) can also be recast in the same way.
	\end{Lemma}
	\begin{proof}
		Define $\mathcal{A} = \log(1+\frac{|\h_k^H\om_k|^2}{\sum_{j\in \kappa_k} |\h_k^H\om_j|+\sigma^2})$, $\mathcal{B} =  \log(1+\frac{|\h_k^H\om_k|^2}{P_t|h_{dk}|^2+\sum_{j\in \kappa_k} |\h_k^H\om_j|+\sigma^2})$ and $\mathcal{R} = R_{min}$. Substitute (\ref{rtktok}) and (\ref{r1tktok}) into (\ref{Rkk>Rmin}). We have 
		\begin{equation}
			\tau \mathcal{A} + (1-\tau) \mathcal{B} \geq  \mathcal{R},
		\end{equation} 
		which is equivalent to 
		\begin{equation} \label{17}
			\tau (\mathcal{A}-\mathcal{R}) \geq  (\tau-1) (\mathcal{B}-\mathcal{R}).
		\end{equation}
		It can be observed that $\tau-1 \leq 0$ and $\mathcal{A} \geq \mathcal{B}$ are always held. Therefore, (\ref{17}) can be always satisfied if $\mathcal{B} \geq \mathcal{R}$ is held. 
	\end{proof}
	According to lemma 1, it can be observed that the constraint (\ref{P1b}) is equivalent to the following two constraints:
	\begin{equation}\label{Rkt>Rmin}
		\frac{1}{1-\tau}R_{k \to t}^{(1-\tau)} \geq R_{min},
	\end{equation}
	\begin{equation} \label{Rkk>Rmin}
		\frac{1}{1-\tau}R_{k \to k}^{(1-\tau)} \geq R_{min}.
	\end{equation}
	In particular, the constraint (\ref{P1b}) is equivalent to the following two inequations:
	\begin{equation}
		\begin{aligned}
			 \log(1+\frac{|\h_t^H\om_k|^2}{P_t|h_{dt}|^2+\sum_{j\in \kappa_k} |\h_t^H\om_j|^2+\Theta_t+\sigma^2})\geq R_{min}
		\end{aligned}		
	\end{equation} 
	and
	\begin{equation}
		\begin{aligned}
			 \log(1+\frac{|\h_k^H\om_k|^2}{P_t|h_{dk}|^2+\sum_{j\in \kappa_k} |\h_k^H\om_j|^2+\Theta_k+\sigma^2})\geq R_{min},
		\end{aligned}
	\end{equation} 
	where $\Theta_k  = \sum_{i=1}^{K}|\bep_k^H\om_i|^2$, $1\leq k \leq K$.
	By defining $\bt \triangleq \frac{\tau}{1-\tau}$, (\ref{Ec}) then can be recast as $E_c = \bt(\eta P_r+P_c)+P_c$. After some simple manipulation, the problem P1 can be further reduced to 
	\begin{subequations} \label{P2}
		\begin{align}
			\textnormal{P2:} &\max_{\{ \bt, \om \}} \quad  \frac{\log(1+\frac{ \bt\eta P_r |h_{dd}|^2}{\sum_{k = 1}^{K} |\h_{D_r}^H\om_k|^2+ P_t |\epsilon_{dd}|^2+\sigma^2} )}{\bt(\eta P_r+P_c) +P_c } \label{P2a}\\
			\mbox{s.t.} \quad & \frac{|\h_t^H\om_k|^2}{\bt \eta P_r |h_{dt}|^2 + \sum_{j\in \kappa_k}|\h_t^H\om_j|^2 +\Theta_t + \sigma^2} \geq \gamma_{min}, \label{P2b} \\ 
			& 1\leq k < t \leq K \nonumber, \nonumber \\
			\quad & \frac{|\h_k^H\om_k|^2}{\bt \eta P_r |h_{dk}|^2 + \sum_{j\in \kappa_k}|\h_k^H\om_j|^2 +\Theta_k +\sigma^2} \geq \gamma_{min}, \label{P2c}\\
			& 1 \leq k \leq K-1, \nonumber \\
			\quad &  \sum_{k=1}^{K} |\om_k|^2 \leq P_{max}, 1\leq k\leq K \label{P2d} \\
			\quad & 	\bt \geq 0, \label{P2e}
		\end{align}
	\end{subequations}
	where $\gamma_{min} = 2^{R_{min}}-1$. P2 is not a concave problem due to the non-concave objective function (\ref{P2a}) and two non-convex constraints (\ref{P2b}) and (\ref{P2c}). Because the time switching dependent variable $\bt$ and the beamforming $\om$ are highly coupled in the problem, P2 is difficult to obtain the optimal solution directly. In the next section, P2 is divided into two subproblems and a fractional programming based alternating algorithm is proposed to iteratively optimize $\bt$ and $\om$. 
	\section{Fractional Programming based Joint Robust Beamforming Design} 
	In this section, an alternating algorithm is proposed to tackle the non-concave problem.  In particular, the problem is divided into two subproblems, one is the time switching coefficient $\tau$ optimization by applying the Dinkelbach method \cite{dinkelbach1967nonlinear}, and the other is robust beamforming vectors designing by applying the complex multi-dimension quadratic transform\cite{shen2018fractional}. Each subproblem is analysed and converted from a non-concave form to a concave form. Following that, these tractable subproblems can be addressed in Matlab using convex optimization tools like CVX and fmincon. 
	\subsection{Time Switching Coefficient Optimization}
	For given beamforming vectors $\om = [\om_1, \cdots, \om_k], 1\leq k \leq K$, the problem P2 can be reduce to
	\begin{subequations} \label{P3}
		\begin{align}
			\textnormal{P3:} \max_{ \bt} \quad & \frac{\log(1+\frac{\bt A}{\bt B+C})}{\bt D + E} \label{P3a} \\
			\mbox{s.t.} \quad & \bt a_t + b_{k,t} \leq 0, 1\leq k < t \leq K, \label{P3b} \\ 		
			\quad &  \bt c_k + d_k \leq 0,  1\leq k\leq K,  \label{P3c}\\
			\quad & 	\bt \geq 0, \label{P3d}
		\end{align}
	\end{subequations}  
	where
	\begin{equation}
		\begin{cases}
			A = \eta |h_{dd}|^2\sum_{k=1}^{K} |\h_{D_t}^H\om_k|^2, & \\
			B = \eta |\epsilon_{dd}|^2\sum_{k=1}^{K} |\h_{D_t}^H\om_k|^2, &\\ 
			C = \sum_{k=1}^{K} |\h_{D_r}^H\om_k|^2 +\sigma^2, &\\
			D = \eta \sum_{k=1}^{K} |\h_{D_t}^H\om_k|^2+P_c, &\\
			E =  P_c, &\\
			a_t =\eta|h_{dt}|^2\sum_{k=1}^{K} |\h_{D_t}^H\om_k|^2, &\\
			b_{k,t} = \sum_{j\in \kappa_k}|\h_t^H\om_j|^2 +\Theta_t+\sigma^2 - \frac{|\h_t^H\om_k|^2}{\gamma_{min}}, &\\
			c_k = \eta |h_{dk}|^2\sum_{k=1}^{K} |\h_{D_t}^H\om_k|^2, &\\
			d_k = \sum_{j\in \kappa_k}|\h_k^H\om_j|^2 +\Theta_k+\sigma^2 - \frac{|\h_k^H\om_k|^2}{\gamma_{min}}.&\\
		\end{cases}
	\end{equation}
	\begin{Lemma}
		(\ref{P3a}) is a concave-convex function of $\bt$, regardless of whether channel estimation errors are existing.
	\end{Lemma}
	\begin{proof}
		\begin{itemize}
			\item[\textit{1)}] \textit{For the case imperfect CSI is obtained:}
		\end{itemize}
		Define $g(\bt) = \log(1+\frac{\bt A}{\bt B+C})$, $h(\bt)$ is a concave function of $\bt$ and its second-order derivative is given by
		\begin{equation}
			\frac{\partial^2g(\bt)}{\partial^2\bt} = -\frac{AC(2B^2\bt+2AB\bt+AC+2BC)}{\ln2((A+B)\bt+c)^2(B\bt+C)^2}.
		\end{equation}
	\item[\textit{2)}] \textit{For the case perfect CSI is obtained:}
		For this case, the numerator of (\ref{P3a}) is reduced to $\hat{g}(\bt) = \log(1+\bt \hat{A})$, where $\hat{A} = \frac{\eta P_r|h_{dd}|^2}{ P_i +\sigma^2}$. The second-order derivative of $\hat{g}(\bt)$ is given by
		\begin{equation}
			\frac{\partial^2\hat{g}(\bt)}{\partial^2\bt} = -\frac{\hat{A}^2}{\ln2 (1+\hat{A})^2}
		\end{equation}
	With non-negative $A$, $B$, $C$ and $\hat{A}$, the second-order derivative of $g(\bt)$ and $\hat{g}(\bt)$ are both non-positive. Therefore, $g(\bt)$ and $\hat{g}(\bt)$ are both concave function of $\bt$.
	\end{proof}
	According to lemma 2 and the fact that the denominator of (\ref{P3a}) is a linear function with respect to $\bt$. It can be observed that (\ref{P3a}) is a single-ratio concave-convex function of $\bt$, and hence Dinkelbach method can be applied to transform it into a concave function \cite{dinkelbach1967nonlinear}. 
	\begin{proposition}
		The maximum EE can be achieved when $F(q^*) = 0$, where $F(q)$ is defined as follows
		\begin{subequations} \label{P4}
			\begin{align}
				\textnormal{P4:} F(q) = \max_{ \bt} \quad & \log(1+\frac{\bt^* A}{\bt^* B+C})-q^*(\bt^*D + E ) \label{P4a} \\
				\mbox{s.t.} \quad & (\ref{P3b})-(\ref{P3d}),
			\end{align}
		\end{subequations} 
	\end{proposition}
	where \begin{equation} \label{q}
		q^* =  \frac{\log(1+\frac{\bt^* A}{\bt^* B+C})}{\bt^* D + E}.
	\end{equation}
	\begin{proof}
		Please refer to \cite{dinkelbach1967nonlinear}.
	\end{proof}
	For a given $q$ the objective function (\ref{P4a}) is a concave function minus a convex function with respect to $\bt$, which yields a concave maximization problem. Therefore, the problem is reduced to a linear constraints concave problem and can be solved by convex optimization tools. In the end, the optimized time switching coefficient $\tau^*$ can be obtained by $\tau ^*= \frac{\bt^*}{1+\bt^*}$.
	
	\subsection{Robust Beamforming Design}
	The last subsection developed the Dinkelbach method to optimize the time switching coefficient $\tau$ under fixed beamforming vectors $\om$. This section focuses on the robust beamforming design by regarding $\tau$ as a constant. 
	
	For a given time switching coefficient $\tau$, the problem P2 can be recast as follows
	\begin{subequations} \label{P5}
		\begin{align}
			\textnormal{P5:} \max_{\{ \om \}}& \quad  \frac{\log(1+\frac{ \bt\eta P_r |h_{dd}|^2  }{P_i+ \Theta_{dd} + \sigma^2} )}{\bt(\eta \sum_{k = 1}^{K} |\h_{D_t}^H\om_k|^2+P_c) +P_c } \label{P5a}\\
			\mbox{s.t.} \quad & \frac{|\h_t^H\om_k|^2}{\bt \eta |h_{dt}|^2 P_r  + \sum_{j\in \kappa_k}|\h_t^H\om_j|^2 + \Theta_t+ \sigma^2 } \geq \gamma_{min},  \label{P5b}\\ \quad & 
			1\leq k < t \leq K, \nonumber \\ 
			\quad & \frac{|\h_k^H\om_k|^2}{\bt \eta |h_{dk}|^2 P_r  + \sum_{j\in \kappa_k}|\h_k^H\om_j|^2 +\Theta_k+\sigma^2}\geq \gamma_{min}, \label{P5c}\\ \quad &  1\leq k \leq K \nonumber \\ 		
			\quad &  \sum_{k=1}^{K} |\om_k|^2 \leq P_{max}, 1\leq k\leq K, \label{P53} 
		\end{align}
	\end{subequations}
	Note that $P_r = \sum_{k = 1}^{K} |\h_{D_t}^H\om_k|^2$, $P_i = \sum_{k = 1}^{K} |\h_{D_r}^H\om_k|^2$ and $\Theta_{dd} = \bt\eta\sum_{k = 1}^{K} |\h_{D_t}^H\om_k|^2 |\epsilon_{dd}|^2$. The problem $P5$ is not a concave optimization problem as the existing of the non-concave objective function (\ref{P5a}) and the two non-convex constraints (\ref{P5b}), (\ref{P5c}). Note that the Dinkelbach's method is no longer applicable as the objective function is not a concave-convex fractional form. In order to transform $P5$ to a tractable convex optimization problem,  the multidimensional and complex quadratic transform \cite{shen2018fractional} is applied. 
	\begin{Lemma}
		For two functions $\bA_m(\om): \mathbb{C}^{d_1} \to \mathbb{C}^{d_2}$ and $\bB_m(\om): \mathbb{C}^{d_1 } \to \mathbb{S_{++}}^{d_2\times d_2}$, $m \in \mathbb{N}_+$, the following equivalent can be established
		\begin{equation} \label{lemma1}
			\begin{aligned}
				&\sum_{m = 1}^{M} \bA_m^H(\om)\bB_m^{-1}(\om)\bA_m(\om)\\ & = \max_{z} \sum_{m=1}^{M} (2\tRe\{z_m^H\bA_m(\om)\} - z_m^H\bB_m(\om)z_m),
			\end{aligned}
		\end{equation}
		where $z_m$ are introduced auxiliary variables.
	\end{Lemma}
	\begin{proof}
		Define $f(z_m) = 2\tRe\{z_m^H\bA_m(\om)\} - z_m^H\bB_m(\om)z_m$. Note that $f(z_m)$ is a linear function minus a quadratic function with respect to $z_m$ (i.e., a concave function in terms of $z_m$). Therefore, the maximum value of $f(z_m)$ can be achieved when $\frac{\partial f(z_m)}{z_m} = 0$. The optimal $z_m^*$ that can maximize $f(z_m)$ is $z_m^* = \bB_m(\om)^{-1}\bA_m(\om)$. Hence, the maximum value of $f(z_m)$ can be obtained by substitute $z_m^*$ into $f(z_m)$ (i.e., $f(z_m^*) = \bA_m(\om)^H\bB_m(\om)\bA_m(\om)$). The equivalence of (\ref{lemma1}) now is established. 
	\end{proof}
	\begin{proposition}
		The objective function (\ref{P5a}) is equivalent to the following concave form in terms of $\om$, \vspace{-0.5cm}
		\begin{equation}\label{fqq}
			\begin{aligned}
				f_{qq} (\bw, y, \bz) = & 2y(\log(1+\bt\eta|h_{dd}|^2  \sum_{k=1}^{K}(2\tRe\{ z_k^H\h_{D_t}^H\om_k\}\\
				&-z_k^H(P_i + \Theta_{dd} + \sigma^2)z_k)))^{\frac{1}{2}}\\
				& - y^2 (\bt\eta P_r +(1+\bt)P_c)
			\end{aligned}
		\end{equation}
		if the introduced auxiliary variables $y$ and $\bz = \{z_1, \cdots, z_k\}$ can satisfy (\ref{opty}) and (\ref{optz}) respectively. 
		\begin{equation} \label{optz}
			z_k^* = \frac{\h_{D_t}^H\om_k}{P_i+ \Theta_{dd}+\sigma^2}, 1 \leq k \leq K.
		\end{equation}
		\begin{equation}\label{opty}
			y^* = \frac{\sqrt{R(\bw)}}{E(\bw)},
		\end{equation}
		where $R(\bw)$ is given by (\ref{rw}) and $E(\bw) = \bt(\eta P_r+P_c) + P_c$. $\bw$ refers to the collection of $\{\om_k\}$. 
	\end{proposition}
	\begin{proof} 
		First, we prove the equivalence between (\ref{P5a}) and (\ref{fqq}).
		In order to decouple the numerator and denominator of (\ref{P5a}), the single-ratio quadratic transform \cite{shen2018fractional} is first applied. 
		\begin{equation} \label{fq}
			f_q (\bw,y) = 2y \hR(\bw)^{\frac{1}{2}} - y^2E(\bw),
		\end{equation}
		where $\hR(\bw) = \log(1+\frac{ \bt\eta P_r |h_{dd}|^2  }{P_i+ \Theta_{dd} + \sigma^2} )$. (\ref{fq}) is equivalent to (\ref{P5a}) if $f_q (\bw,y)$ can achieve the maximum value with optimal $y^*$. The first-order derivative of (\ref{fq}) with respect to $y$ is $\frac{\partial f_q(y))}{\partial y} = 2\sqrt{\hR(\bw)} - 2yE(\bw)$. Since (\ref{fq}) is a quadratic function of $y$, its optimal $y^*$ can be obtain by letting $\frac{\partial f_q(y)}{\partial y} =0$, which yields (\ref{opty}). Substitute (\ref{opty}) into (\ref{fq}) the objective function (\ref{P5a}) is retrieved and the equivalence is established. 
		Although $-E(\bw)$ is concave due to its minus quadratic form with respect $\om$, (\ref{fq}) is still non-concave in terms of $\bw$, because the concavity of $\hR(\bw)$ is unprovable. To restore the concavity of $
		R(\bw)$, lemma 3 is applied. Thus, $\hR(\bw)$ can be recast to
		\begin{equation} \label{rw}
			\begin{aligned}
				R(\bw) = & \log(1+\bt\eta|h_{dd}|^2 \sum_{k=1}^{K} (2\tRe\{z_k^Hh^H_{D_t}\om_k\}\\
				&-z_k^H\sum_{k = 1}^{K} |\h_{D_r}^H\om_k|^2z_k)),
			\end{aligned}
		\end{equation}
		Similarly, (\ref{rw}) is equivalent to the numerator of (\ref{P5a}) when (\ref{rw}) can achieve its maximum value with optimal $\bz^*$ (i.e., when $z_k$ satisfies (\ref{optz})), where $\bz$ denotes the collection of $z_k$. By combining (\ref{fqq}), (\ref{optz}), (\ref{opty}), (\ref{fq}) and (\ref{rw}), the equivalence from (\ref{fqq}) to (\ref{P5a}) is proved.
		 
		Second, we show the concavity of (\ref{fqq}) in terms of $\bw$. To show (\ref{fqq}) is a concave function of $\bw$, it is sufficient to prove (\ref{rw}) is concave. Define $f(\bw) = \sum_{k=1}^{K} (2\tRe\{z_k^Hh^H_{D_t}\om_k\}-z_k^H\sum_{k = 1}^{K} |\h_{D_r}^H\om_k|^2z_k)$, its Hessian matrix is given by
		\begin{equation}
			\nabla^2f(\bw) = - 2\begin{bmatrix}
				 {z_1^*}^2\bm{H}_{D_r} & \bm{0} & \cdots & \bm{0}\\
				\bm{0} & {z_2^*}^2\bm{H}_{D_r}  & \quad & \vdots\\
				\vdots & \quad & \ddots & \bm{0}\\
				\bm{0} & \cdots & \bm{0} &{z_k^*}^2\bm{H}_{D_r} 
			\end{bmatrix}.
		\end{equation}
		Note that $\bm{H}_{D_r} = \h_{D_r}\h_{D_r}^H$ is a positive semidefinite matrix and ${z_k^*}^2 \geq 0$, therefore, $\nabla^2f(\bw)$ is negative semidefinite. As a result, $f(\bw)$ is a concave function. On the other hand, it can be observed that $\log(x), x \geq 0$ is non-decreasing and concave. According to \textit{operations that preserve convexity} for vector composition in \cite{boyd2004convex}, it can be proved that (\ref{rw}) is concave. Furthermore, it is worth to mention that the function $f(x) = x^{\frac{1}{2}}$ is also concave and non-decreasing. Hence, the first term of (\ref{fqq}) is concave with respect to $\bw$. Due to the negative quadratic form of $\om$, the second term of (\ref{fqq}) is also concave. The concavity of (\ref{fqq}) in terms of $\bw$ has been proved. 
	\end{proof}
	Similarly, Lemma 3 is also applicable to attain the concavity of constraints (\ref{P5b}) and (\ref{P5c}). They can be transformed into two concave sets with $\om$ which are shown as follows:   
	\begin{equation}\label{qt1}
		\begin{aligned} 
			&\frac{|\h_t^H\om_k|^2}{\bt \eta |h_{dt}|^2 P_r  + \sum_{j\in \kappa_k}|\h_t^H\om_j|^2 + \Theta_t+ \sigma^2 }\\   =&\max_{ \nu_{t,k}} 2\tRe\{\nu_{t,k}^H h_t^H \om_k\} - \nu_{t,k}^H \alpha_{t,k}\nu_{t,k} \geq  \gamma_{min}, 
		\end{aligned}
	\end{equation}
	\begin{equation}\label{qt2}
		\begin{aligned}
			&\frac{|\h_k^H\om_k|^2}{\bt \eta |h_{dk}|^2 P_r  + \sum_{j\in \kappa_k}|\h_k^H\om_j|^2 +\Theta_k+ \sigma^2}\\  =& \max_{ \mu_k} 2\tRe\{\mu_k^H h_k^H \om_k\} - \mu_k^H \beta_k\mu_k \geq \gamma_{min}, 
		\end{aligned}
	\end{equation}
	where $\alpha_{t,k} = |h_{dt}|^2\bt \eta P_r + \sum_{j\in \kappa_k} |h_t^H\om_j|^2+\Theta_k+\sigma^2$ and $\beta_k = |h^H_{dk}|^2 \bt \eta P_r + \sum_{j\in \kappa_k}|h_k^H\om_j|^2+\sigma^2$. $\nu_{t,k}$ and $\mu_k$ are two introduced auxiliary variables which can be updated by
	\begin{equation}\label{nu}
		\nu_{t,k}^* = \frac{h_t^H\om_k}{\alpha_{t,k}}, (1
		\leq t < k \leq K),
	\end{equation} 
	\begin{equation}\label{mu}
		\mu_k^* = \frac{h_k^H\om_k}{\beta_k}, (1\leq k \leq K).
	\end{equation} 
	Denote the collection of $\{\nu_{t,k}\}$ and $\{\mu_k\}$ by $\bm{\nu}$ and $\bm{\mu}$ respectively. By using (\ref{fqq}), (\ref{optz}), (\ref{opty}), (\ref{qt1}) and (\ref{qt2}), the problem $P5$ can be reformulated as
	\begin{subequations} \label{P6}
		\begin{align}
			\textnormal{P6:} &\max_{\{ \bw, y, \bz, \bm{\nu}, \bm{\mu} \}} \quad  f_{qq}(\bw, y, \bz ) \label{P6a}\\
			\mbox{s.t.} \quad & 2\tRe\{\nu_{t,k}^H \h_t^H \om_k\} - \nu_{t,k}^H \alpha_{t,k}\nu_{t,k} \geq \gamma_{min}, 1\leq t < k \leq K,\label{P6b}\\ 
			\quad & 2\tRe\{\mu_k^H \h_k^H \om_k\} - \mu_k^H \beta_k\mu_k \geq \gamma_{min}, 1\leq k \leq K-1 \label{P6c}  \\
			\quad &  \sum_{k=1}^{K} |\om_k|^2 \leq P_{max}, 1\leq k\leq K. \label{P6e} \\
			\quad & (\ref{opty}), (\ref{optz}), (\ref{nu}),(\ref{mu}). \label{P6f}
		\end{align}
	\end{subequations}
	For given $y$, $\bz$, $\bm{\nu}$, and $\bm{\mu}$, the (\ref{P6a}) is a concave function and constraints (\ref{P6b}) - (\ref{P6f}) are all convex set in regard to $\om$. Hence, problem P6 is a concave optimization problem \cite{boyd2004convex}, and therefore can be solved by convex optimization tools such as CVX or Matlab fmincon. The original con-concave energy efficiency maximization problem P1 has been solved by tackling the subproblems P4 and P6 alternately. The quadratic transform and Dinkelbach method based alternating algorithm, namely, pure fractional programming (PFP) is proposed to maximize the energy efficiency of WPT-assisted D2D communications in MISO-NOMA downlink networks. The proposed algorithm is summarised in Algorithm \ref{alg:1}. 
	\begin{algorithm}[t]
		\caption{Proposed PFP algorithm}
		\label{alg:1}
		\begin{algorithmic}[1]
			\STATE \textbf{Initialization:} Initialize $\bw$ and $\tau$ to a feasible value.
			\REPEAT
			\STATE Update $z_k$ by using (\ref{optz}).
			\STATE Update $y$ by using (\ref{opty}). 
			\STATE Update $\nu_{t,k}$ by using (\ref{nu}).
			\STATE Update $\mu_k$ by using (\ref{mu}).
			\STATE With fixed $z_k$, $y$, $\nu_{t,k}$ and $\mu_k$, solve the problem P6 and obtain the optimized $\bw$.
			\STATE With optimized $\bw$, update $q$ by using (\ref{q}).
			\STATE With fixed $q$, solve the problem P4 and obtain the optimized $\bt$
			\UNTIL The value of (\ref{fqq}) is convergent. 
		\end{algorithmic}  
	\end{algorithm}
	\subsection{A Partial Exhaustive Search based Algorithm}
	To further demonstrate the optimality of the proposed algorithm, this subsection provides a partial exhaustive search based algorithm to optimize $\tau$ for comparison. As discussed in the last subsection, the energy maximization problem can be transformed into a concave problem with respect to beamforming vectors $\om$ for a given time-switching coefficient $\tau$. Therefore, the solution can be obtained by solving P6 for all $\tau$ and selecting the one that corresponds to the maximum energy efficiency. The partial exhaustive search algorithm is summarised in Algorithm \ref{alg:2}. 
	\begin{algorithm}[t]
		\caption{Partial Exhaustive Search (PES) for $\tau$}
		\label{alg:2}
		\begin{algorithmic}[1]
			\STATE \textbf{Initialization:} Initialize $\bw$ and $\tau$ to a feasible value. Initialize the step size $\xi$.
			\FOR {$\tau$ = 0.001 : $\xi$ : 0.999}
			\REPEAT
			\STATE Update $z_k$ by using (\ref{optz}).
			\STATE Update $y$ by using (\ref{opty}).
			\STATE Update $\nu_{t,k}$ by using (\ref{nu}).
			\STATE Update $\mu_k$ by using (\ref{mu}).
			\STATE With fixed $z_k$, $y$, $\nu_{t,k}$ and  $\mu_k$, solve the problem P6 and obtain the optimized $\bw$.
			\UNTIL{The value of (\ref{fqq}) is convergent.}
			\ENDFOR
			\STATE Select the $\tau$ corresponding to the maximum (\ref{fqq}).
		\end{algorithmic}  
	\end{algorithm}
	\section{A reinforcement learning based approach to maximize the Energy Efficiency}
	In this section, a reinforcement learning based algorithm, deep deterministic policy gradient (DDPG), is first introduced. Afterwards, the structure of neural networks and training procedures are provided. At the end of this section, we discuss the application of DDPG to the proposed problem including the setup of action, state and reward, as well as constraints handling.
	\subsection{A Brief Introduction to DDPG:} 
	Reinforcement learning (RL) is neither like supervised learning uses an external supervisor labelled data set to learn, nor like unsupervised learning which aims to find the hidden structure in the unlabelled collections\cite{sutton2018reinforcement}, RL learns through the way that letting the agent interacts with the environment. Specifically, in RL, the agent decides what actions should be taken according to the current observation (also termed state) and then obtains the corresponding reward. Macroscopically, RL aims to find an optimal action that maximizes the reward. RL can generally be divided into two types which are value-based and policy-based respectively. Q-learning and state-action-reward-state'-action' (SARSA) are two typical value-based learning. They only solve the problem with low-dimension discrete actions. Policy gradient (PG), as a policy-based RL, can solve the problems with continuous actions. However, PG usually convergents at a local optimal and evaluates a policy inefficiently. The combination of Q-learning and deep neural networks (DNN) derives deep Q network (DQN), which is applicable to the problems when the state space and the discrete action space are enormous. To handle the problems with high-dimension continuous actions, DDPG is proposed by integrating DQN and PG \cite{2015Continuous}. By the fact that beamforming vectors and the time coefficient are high-dimension and continuous variables, DDPG is selected to solve the problem in this section.   
	\subsection{Exploration and Experience Replay}
	In DRL, inspired by the greedy strategy, noise is added to the actor network's output to encourage the agent to explore the surroundings, therefore the action to be taken for state $s^{(t)}$ is determined by \cite{2015Continuous}
	\begin{equation}
		a_t = \mu(s_t|\theta_{\mu'}) + \mathcal{N}_t.
	\end{equation}
	The same as DQN, DDPG applies experience replay to improve the training efficiency as well \cite{schaul2015prioritized}. To be more precise, all transitions $\{a,s,r,s'\}$ are first stored into the experience replay buffer and the training starts if the buffer is saturated. Afterwards, a certain number of transitions (also called mini-batch) are randomly selected to train those networks. 
	\subsection{Training Neural Networks:} 
	Different from other RL methods, DDPG has four neural networks: 
	\begin{itemize}
		\item An actor network $\mu(s|\theta_\mu)$: input current state $s$ then output action $a$.
		\item A critic network $Q(s,a|\theta_q)$: input current $a$ and $s$ then output Q-value.
		\item A target actor network $\mu'(s'|\theta_{\mu'})$: input state $s'$ then output target action $a'$.
		\item A target critic network $Q'(s',a'|\theta_{q'})$: input $a'$ and $s'$ then output target Q-value.
	\end{itemize}
	$\theta$ represents the parameters of the corresponding neural network. The mathematical expression of the training process is as follows. To train the actor network, the gradient ascend and chain rule are used for the Q-value function
	\begin{equation} \label{Qupdata}
		\nabla_{\theta_{\mu}}J =  \frac{1}{N_B} \sum_{t=1}^{N_B} (\nabla_aQ(s_t, \mu(s_t|\theta_{\mu})|\theta_q) \nabla_{\theta_{\mu}}\mu(s_t|\theta_{\mu})),
	\end{equation}
	where $N_B$ is the size of mini-batch. The critic network is trained by minimizing the loss between the current Q-value and target state-value 
	\begin{equation} \label{L}
		L = (y - Q(s,a|\theta_q))^2,
	\end{equation}
	where $y$ is the target value for the previous state-value which is given by
	\begin{equation} \label{y}
		y = r + \gamma Q'(s',\mu'(s'|\theta_{\mu'})|\theta_{q'}).
	\end{equation}
	$r$ represents the reward and $\gamma$ denotes the discount factor. It is worth to mention that the two target networks (target actor network and target critic network) have the same framework as their counterparts but the parameters update strategy is different. The two target networks' parameters $\theta_\mu'$ and $\theta_q'$ are updated by soft updating strategy: 
	\begin{equation} \label{soft}
		\theta_\mu' = \xi \theta_\mu + (1-\xi)\theta_\mu', \quad \theta_q' = \xi \theta_q + (1-\xi)\theta_q'.
	\end{equation}    
	where $\xi$ denotes the soft updating coefficient.
	\subsection{Application DDPG to the Problem:} 
	In this paper, the original problem P1 is solved directly\footnote{There is no further processing to the variables highly coupled non-convex problem including the objective function and constraints.} by DDPG. Suppose that the BS is the agent and it can observe the CSI and downlink users' data rate.\\
	\begin{itemize}
	\item[\textit{1)}] \textit{Action Space:} As the optimization needed variables, beamforming and time switching coefficient are naturally defined as the action. Note that all elements of beamforming vectors are complex numbers and the input vectors of neural networks should be real numbers. Hence, we need to split beamforming vectors into real parts and imaginary parts. The action at the $t$-th training step is given by
	\begin{equation}
		\begin{aligned} \label{at}
			a_t = &\left[\right. \bt^{(t)}, \textnormal{Re}\{\om_1^{(t)}\}, \\& \cdots, \textnormal{Re}\{\om_k^{(t)}\}, \textnormal{Im}\{\om_1^{(t)}\}, \cdots, \textnormal{Im}\{\om_k^{(t)}\} \left.\right].
		\end{aligned}
	\end{equation}
	\item[\textit{2)}] \textit{State Space:} The state vector is designed to represent as much information as possible about the current environment and the impact of the action on the system. For the proposed optimization problem, the state vector should include all CSI and all NOMA downlink users' data rates. The state vector at the $t$-th training step is defined as follows: 
	\begin{equation} \label{state}
		\begin{aligned}
			s_t=\left[\right. & |\h_1^{(t)}|^2, \cdots, |\h_k^{(t)}|^2, |\h_{D_t}^{(t)}|^2, |\h_{D_r}^{(t)}|^2, |h_{dd}^{(t)}|^2, \\
			&|\h_{d1}^{(t)}|^2, \cdots, |\h_{dk}^{(t)}|^2, R_1^{(t)}, \cdots, R_k^{(t)}, \mathcal{R}^{(t)}_{t,k},\\& |\om_1^{(t)}|^2, \cdots, |\om_k^{(t)}|^2 \left.\right],
		\end{aligned}
	\end{equation}
	where $\mathcal{R}_{t,k}^{(t)}$ denotes the collection of the data rate of user $t$ to decode user $k$, $1\leq k < t \leq K$.\\
	\item[\textit{3)}] \textit{Reward:} Our aim is to maximize the energy efficiency which can fit the goal of the DDPG algorithm is to maximize the reward. Therefore, the objective function (\ref{P1a}) is naturally defined as the reward.  
	\begin{equation} \label{reward}
		r_t = \left(\frac{R_D}{E_c}\right)^{(t)}
	\end{equation}
	\item[\textit{4)}] \textit{Constraint Handling:} For optimization problems, it is necessary to make optimized variables satisfy all constraints. Terms to DRL, it means the actions of the agent is needed to be restricted in a perspective region. How to let output actions efficiently satisfy constraints is a very important problem. To the best of our knowledge, the simple and brutal punishment mechanism is not an efficient method. Therefore, we combine the punishment mechanism and normalization processing to guarantee that all constraints can be satisfied. In particular, for (\ref{P1a}), we rewrite the reward (\ref{reward}) as follows:
	\begin{equation} \label{newreward}
		r_t = 
		\begin{cases}
			\left(\frac{R_D}{E_c}\right)^{(t)} & \forall R_{k \to t} \geq R_{min}\\
			-\zeta |R_{k \to t}^{(t)} - R_{min}| & \exists R_{k \to t}^{(t)} < R_{min}, 
		\end{cases}
	\end{equation}
	where $\zeta$ is the punishment factor. As can be seen, if all NOMA downlink users' data rate at $t$-th step can achieve the minimum target rate, the agent obtains a positive reward, otherwise, the agent is punished by a negative value which depends on how bad the action is. Inspired by the fact that all the PFP optimized beamforming vectors meet the equivalence of (\ref{P1c}), we apply normalization to the output beam vectors in each step to guarantee the power constraint (\ref{P1c}) can be guaranteed. At $t$-th training step, the normalized beamforming vectors can be represented as:  
	\begin{equation}
		\hat{\om}_k^{(t)} = \sqrt{\rho_k^{(t)}}\frac{\om_k^{(t)}}{|\om_k^{(t)}|^2},
	\end{equation}
	where $\rho_k^{(t)}$ is the power allocation coefficient and $\om_k^{(t)}$ represents the beamforming vectors outputted by action network. $\rho_k^{(t)}$ is given by
	\begin{equation}
		\rho_k^{(t)} = P_{max}\frac{|\om_k^{(t)}|^2}{\sum_{k=1}^{K}|\om_k^{(t)}|^2}.
	\end{equation}
	By this normalizing, the summation of all new beamforming vectors can always meet $\sum_{k=1}^{K}\hat{\om}_k^{(t)} = P_{max}$, and hence the (\ref{P1c}) is guaranteed. Meanwhile, $\hat{\om}_k^{(t)}$ remains the same direction with $\om_k^{(t)}$. For (\ref{P1d}), because the time switching coefficient $\tau$ has been converted to $\bt$ whose feasible range is synchronously shifted, we only need to map the first element of (\ref{at}) to the non-negative field by using some functions, such as abs($x$).
	\end{itemize}
	The detail of the DDPG algorithm is shown in Algorithm \ref{alg:3} and the framework of neural networks with their parameters setup is provided in the simulation section. 
		\begin{algorithm}[t]		
		\caption{DDPG-based algorithm}
		\begin{algorithmic}[1]\label{alg:3}
			\STATE \textbf{Initialization:} Randomly initialize the critic evaluation network $\theta_q$ and the actor evaluation network $\theta_{\mu}$. Initialize the critic target network $\theta_{q'} = \theta_q$ and the actor target network $\theta_{\mu'}=\theta_{\mu}$.
			
			Initialize the experience replay buffer $\mathcal{D}$ with capacity $\mathcal{C}$.
			
			Initialize the learning rate $\beta$, the discount factor $\lambda$, the soft update coefficient $\xi$ and the minibatch size $N_B$. 
			\FOR{episode $j = 1, \cdots, J$}
			\STATE Randomly initialize the time switching coefficient $\tau$ the beamforming vectors $\om^{(j)}$.
			\STATE Decide the NOMA downlink users' decoding order according to current channel gains.
			\STATE Obtain the initial state $s_1$ (\ref{state}).
			\FOR{step $t = 1, \cdots,T$}
			\STATE Initialize the random process $\mathcal{N}$ for action exploration.
			\STATE Choose action $a_t = \mu(s_{t-1}|\theta_{\mu}) + \mathcal{N}_t$.
			\STATE Extract corresponding actions to retrieve beamforming vectors and normalize them.
			\STATE Obtain the current state $s_t$.
			\STATE Set $r_t$ according to (\ref{newreward}).
			\STATE Store transition $\{s_t,a_t,r_t,s_{t+1}\}$ into the replay buffer $\mathcal{D}$.
			\STATE Sample $N_B$ minibatch transitions from $\mathcal{D}$ to train.
			\STATE Calculate target Q value by the equation (\ref{y}).
			\STATE Update the critic evaluation network $Q(s,a|\theta_q)$ by minimizing the loss function (\ref{L}).
			\STATE Update the actor evaluation network $\mu(s|\theta_{\mu})$ by using the sampled policy gradient in (\ref{Qupdata}) .
			\STATE Update two target networks by using soft update (i.e.(\ref{soft})).
			\STATE Transfer state $s_t$ to $s_{t+1}$.
			\ENDFOR
			\ENDFOR
		\end{algorithmic}  
	\end{algorithm}
	\begin{Remark}
		As we mentioned in section II, the solution obtained via PFP and DDPG for the considered system model can also be applied to the legacy user-clustered hybrid NOMA downlink networks \cite{you2020user}. For example, if multiple D2D pairs are added into a user-clustered hybrid NOMA downlink system, where one D2D pair is assigned to each cluster, the solution that this paper provided can be used for each cluster. If multiple D2D transmitters and a single receiver are deployed, where one transmitter is assigned to each cluster, the proposed algorithms are still applicable and the D2D communication can be viewed as NOMA uplink.
	\end{Remark}
	\section{Simulation results}
	In this section, we study the performance of the proposed algorithm and DDPG-based algorithm to maximize the energy efficiency of WPT-assisted D2D communications in MISO-NOMA downlink networks. The simulation results for the same communication scenario in MISO-OMA networks are also provided. 
	\subsection{Deep Neural Networks Parameters and Structure Setup}
	Fully connected neural networks are used for both actor networks and critic networks. In actor networks, one input layer, two hidden layers and one output layer are employed, where the rectified linear activation function (ReLU) is used after the first hidden layer and the hyperbolic tangent function (tanh) is used for both the second hidden layer and the output layer. Due to there are two inputs (i.e., state and action) in critic networks, two parallel individual hidden layers are also needed to receive the two input layers' output. Then, the two outputs are concatenated and connected to another hidden layer. In critic networks, only ReLU is applied after all the layers. The number of neurons is 500 for both actor networks and critic networks. Batch normalization is also used for both actor and critic networks with $N_B = 32$ to improve the training performance. Adam optimizer is selected and the learning rate is set to 0.001 for actor networks and 0.002 for critic networks. The soft update coefficient is set to 0.01.
	\subsection{Hyper Parameters Setup}
	In simulations, we assume that the positions of all D2D devices and downlink users are randomly distributed with the region $x,y \in [3,8]$. The BS is deployed at $(0,0)$. Channels are assumed to be the Rayleigh fading and the path loss is also considered. Therefore, the channels can be expressed as
	\begin{equation}
		\h_{sim} = \frac{\h_{Ray}}{\sqrt{d^\alpha}},
	\end{equation}
	where $h_{Ray}$ represents the Rayleigh channel vector, and $d$ and $\alpha$ are the corresponding distance and path loss coefficients, respectively. We set the path loss coefficient between the BS and downlink users to $\alpha_0 = 2.5$ and the path loss coefficient between $D_t$ and $D_r$ is $\alpha_1 = 2$. $\alpha_2 = \alpha _3 = 3.5$ are the path loss coefficient between the BS and $D_t$, and between $D_t$ and downlink users, respectively. For all simulations, the noise power is set to $\sigma^2 = -94$ dBm and the RF energy conversion coefficient is $\eta = 0.1$.
	\subsection{Simulation Results Demonstration}
	\begin{figure}[t]  
		\centering
		\includegraphics[width=1\linewidth]{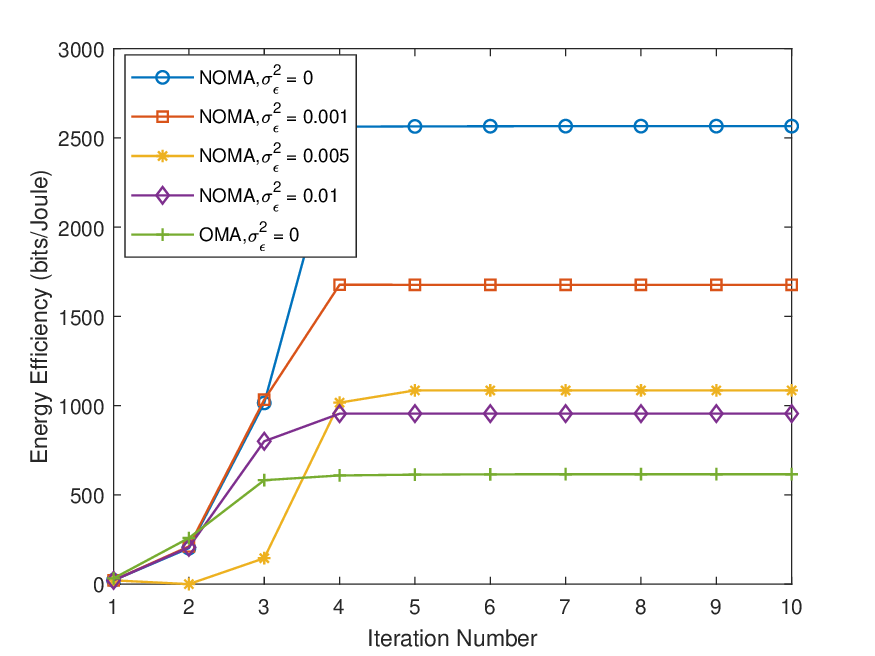}\\
		\caption{Convergence of PFP algorithm for NOMA and OMA with different channel estimation errors. $P_{max} = 20$ dBm, K = 3, M = 16 and $R_{min} = 0.1$ bps/Hz.}\label{EE_vs_iter}
	\end{figure} 
	Fig. \ref{EE_vs_iter} shows the convergence of the proposed algorithm for both the NOMA and OMA schemes, of which the transmit power $P_{max} = 20$ dBm, number of antennas $M = 16$ and number of downlink users $K = 3$. It can be observed that the proposed algorithm can converge very fast for both NOMA and OMA schemes, regardless of whether the channel estimation error is existing. In particular, the maximum value of energy efficiency can be achieved within 5 iterations. On the other hand, this figure preliminary demonstrates the superiority of NOMA. 
	\begin{figure}[t]  
		\centering
		\includegraphics[width=1\linewidth]{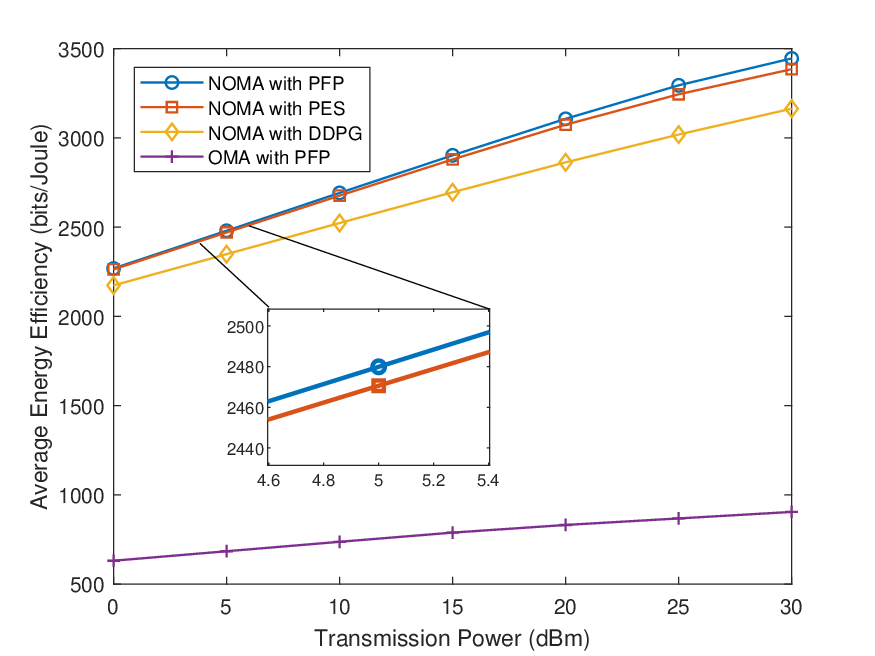}\\
		\caption{Performance demonstration for different algorithms and multiple access schemes, where perfect CSI is obtained. K = 3, M = 16 and $R_{min} = 0.1$ bps/Hz.}\label{Power_vs_EE}
	\end{figure} 

	Fig. \ref{Power_vs_EE} presents the energy efficiency versus transmit power of the BS by applying different algorithms and multiple access schemes. In this figure, we assume that the perfect CSI can be obtained (i.e.,$\sigma_{\epsilon}$ = 0). The number of users and antennas are set to $K = 3$ and $M = 16$, respectively. The minimum target data rate is set as $R_{min} = 0.1$ bps/Hz. The randomness caused by the randomly generated positions and channels is averaged by performing Monte Carlo simulations. It can be observed that the energy efficiency of the WPT supported D2D pair increases with the increase of the BS's transmit power for all algorithms. However, the performances that different schemes can provide are significantly different. In this simulation, we choose the step size of exhaustive search $\xi = 0.1$. It can be seen that the performance of the partial exhaustive search for $\tau$ is slightly worse than the proposed algorithm, which further verified the optimality of the proposed PFP algorithm. The gap between the partial exhaustive search and the proposed algorithm becomes larger when the transmit power increases. This figure also shows the comparison between the proposed PFP and the DDPG-based optimization. It can be observed that, with perfect CSI, the proposed PFP algorithm outperforms the DDPG-based algorithm. Furthermore, Fig. \ref{Power_vs_EE} shows that under the same algorithm optimisation and network framework, the energy efficiency performance of the WPT-assisted D2D communication in the MISO-NOMA system outperforms in the MISO-OMA system significantly. This benefits from the characteristic of the NOMA system that allows all communication resources to be shared. 
	\begin{figure}[t]  
		\centering
		\includegraphics[width=1\linewidth]{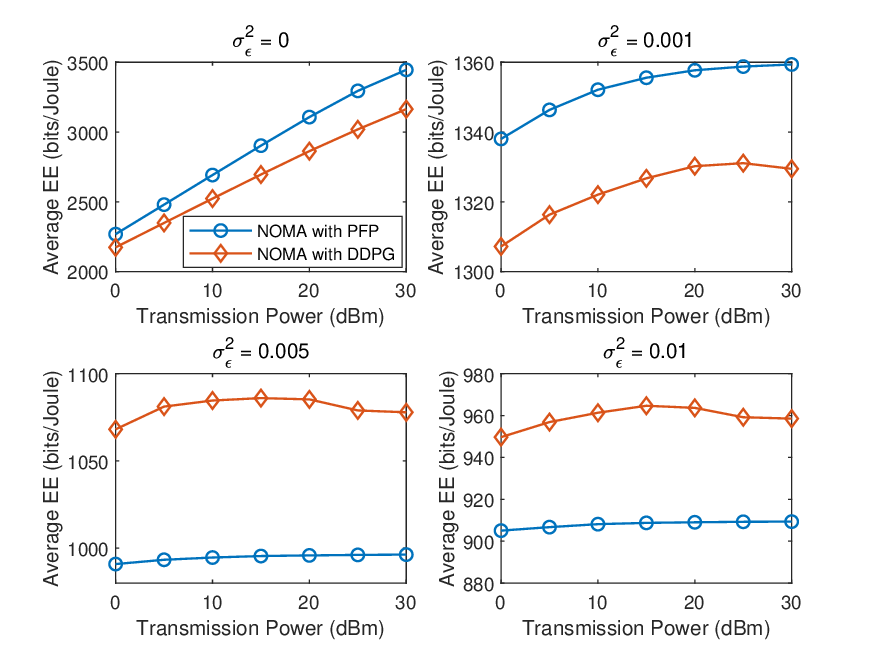}\\
		\caption{Performance comparison between PFP and DDPG for different channel estimation accuracy, where $K = 3$, $M = 16$ and $R_{min} = 0.1$ bps/Hz.} \label{imperfect CSI}
	\end{figure}

	Fig. \ref{imperfect CSI}  illustrates the different performances that can provide by the proposed PFP algorithm and DDPG when the channel estimation accuracy is various. In these simulations, we set $K = 3$, $M = 16$ and $R_{min} = 0.1$ bps/Hz. Channels (i.e., $\h_k$ and $h_{dd}$)  and channel estimation errors (i.e., $\be_k$ and $\epsilon_{dd}$) are used the same for both the proposed algorithm and DDPG. An interesting and important observation is that if channel estimation is perfect or only has slight errors (i.e., $\sigma_{\epsilon} = 0$ and $\sigma_{\epsilon} = 0.001$), the proposed algorithm outperforms DDPG-based algorithm, however, when channel estimation error is severe (i.e., $\sigma_{\epsilon} = 0.005$ and $\sigma_{\epsilon} = 0.01$), the DDPG-based algorithm can provide its better robustness to mitigate the channel estimation error caused performance degradation. On the other hand, unlike the scenario where channel estimation is perfect, the slope decreases with the increase of channel estimation error. This is because when the BS's transmit power increases, the power of channel estimation error caused interference increases as well. Therefore, communication resources will need to be tilted toward downlink users more to guarantee their QoS, which hinders the improvement of the energy efficiency of the D2D pair.
	\begin{figure}[t]  
		\centering
		\includegraphics[width=1\linewidth]{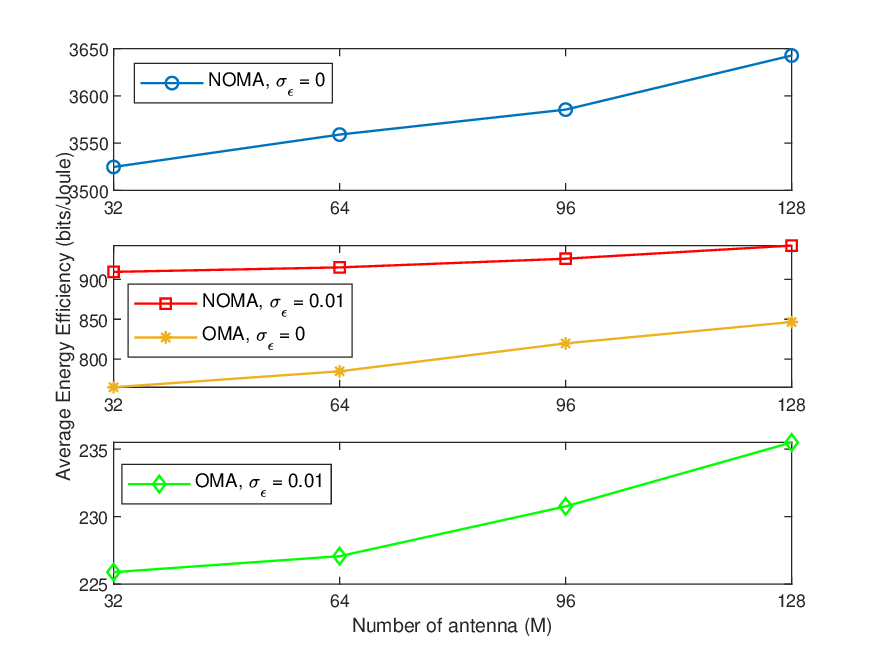}\\
		\caption{Impact of the number of antenna on the performance of the proposed PFP algorithm with NOMA and OMA in different channel estimation errors. K = 3 $P_{max} = 30$ dBm and $R_{min} = 0.1$ bps/Hz}\label{Power_vs_EE(diffM)}
	\end{figure}
	
	Fig. \ref{Power_vs_EE(diffM)} shows the energy efficiency versus the number of antennas. In this simulation, the parameters are set as follows: $K = 4$ and $R_{min} = 0.1$ bps/Hz. To clearly demonstrate differences between different simulations, we plot curves separately into three sub-figures. It can be observed that the NOMA scheme significantly outperforms OMA scheme, even severe channel estimation error is introduced. Benefiting from the spatial diversity, deploying more antennas results in higher D2D pair energy efficiency. However, simply increasing the number of antennas is might not a wise and efficient scheme to improve energy efficiency. Hence, the trade-off between cost and performance improvement is crucial and needs to be considered when designing the system.
	\begin{figure}[t]  
		\centering
		\includegraphics[width=1\linewidth]{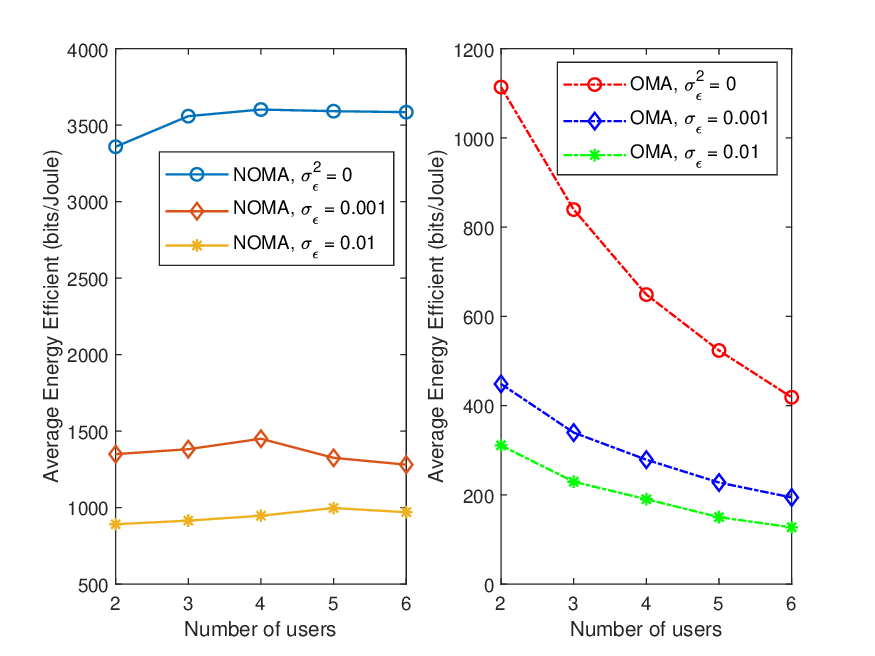}\\
		\caption{Impact of the number of downlink users on the performance of the proposed algorithm with NOMA and OMA in different channel estimation errors, where M = 64 $P_{max} = 30$ dBm and $R_{min} = 0.1$ bps/Hz}\label{KvsEE}
	\end{figure}

	Fig. \ref{KvsEE} illustrates how the number of downlink users impacts the energy efficiency performance of WPT-assisted D2D devices. In this simulation, the number of antennas is M = 64, transmit power is $P_{max} = 30$ dBm and the minimum QoS is $R_{min} = 0.1$ bps/Hz. As can be seen, for the NOMA scheme, with the increase of the number of users the average energy efficiency slightly increases first and then decrease, which can be observed for all $\sigma_{\epsilon}^2$ cases. It is worth to point out that the slight decrease of the case $\sigma_{\epsilon}^2 = 0$ starts from K = 4. Due to the plotting scale, the decrease is not obvious. On the other hand, for OMA scheme, dramatic energy efficiency degradation occurs in all cases. This is because, in OMA networks, more OMA devices result in fewer resources being allocated to each device.
	\begin{figure}[t]  
		\centering
		\includegraphics[width=1\linewidth]{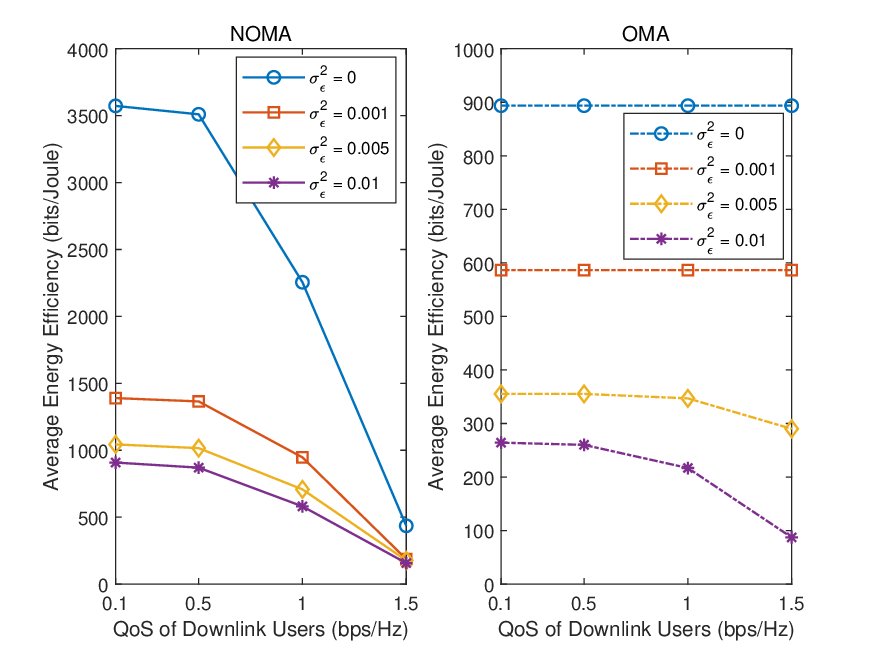}\\
		\caption{Impact of the downlink users' QoS on the performance of the proposed algorithm with NOMA and OMA in different channel estimation errors, where K = 3, M = 64 and $P_{max} = 30$ dBm}\label{RminvsEE}
	\end{figure}
	
	In Fig. \ref{RminvsEE}, the impact of increased downlink users' QoS on the average energy efficiency of WPT-assisted D2D communication is studied. In the MISO-NOMA downlink network, if downlink users' require higher QoS, the energy efficiency of WPT-assisted D2D communications will significantly decrease and the more accurate the channel estimation, the higher the percentage of performance degradation. However, for MISO-OMA downlink networks, the increase of downlink users' QoS impacts the energy efficiency performance of WPT-assisted D2D communications slightly. In particular, when channel estimation is perfect or only has small errors (e.g., $\sigma_{\epsilon}^2 = 0$ or $0.001$), there is no performance degradation if the minimum target rate required by OMA downlink users is increased from 0.1 bps/Hz to 1.5 bps/Hz. When channel estimation is not accurate (i.e.,$\sigma_{\epsilon}^2 = 0.005$ or $0.01$) the increase of OMA downlink users QoS also decreases the energy efficiency of the WPT-assisted D2D communication. Overall, although the increased downlink users' QoS severely affects the WPT-assisted D2D communication in the NOMA downlink network, the NOMA scheme is still a better choice.
	\section{Conclusion}
	In this paper, the proposed PFP algorithm and DDPG-based algorithm are both applied to do the joint robust beamforming design for the WPT-assisted D2D communication in MISO-NOMA downlink networks. The goal is to maximize the energy efficiency of the WPT-enabled D2D devices. To solve the proposed non-concave optimization problem, the PFP algorithm has been proposed to alternatively optimize the beamforming vectors and time switching coefficient. Furthermore, a partial exhaustive search based algorithm has been proposed to prove the PFP algorithm's optimality. The DDPG-based algorithm was performed directly to solve the proposed non-concave problem. Simulations were carried out for both NOMA and OMA schemes with different channel estimation accuracy. In the considered communication scenario, one can conclude is WPT-assisted D2D communication can provide higher energy efficiency if the NOMA scheme is adopted. Another intriguing and important finding is that the proposed PFP algorithm is superior to the DDPG-based algorithm when the perfect CSI can be obtained or just minor errors exists. However, when channel estimation is unsatisfactory, the DDPG-based algorithm is more robust than the PFP algorithm. Based on the finding above, investigating deeper into the causes will be an important direction for our future work. Furthermore, linear energy harvesting has been assumed, and the impact of the energy used for radio frequency circuits and signal processing has not been considered, which prompts us to gain more insight into a more practical model. Alternatively, for realizing battery-less D2D communication, backscattering communication (BackCom) is another mature and efficient scheme that will be another direction for future research.

	\vspace{12pt}
	\bibliographystyle{IEEEtran}
	\bibliography{myrefsletter.bib}
\end{document}